\documentclass[10pt,reqno]{amsart}
\RequirePackage[OT1]{fontenc}
\RequirePackage{amsthm,amsmath}
\RequirePackage[numbers]{natbib}
\RequirePackage[colorlinks,linkcolor=blue,citecolor=blue,urlcolor=blue]{hyperref}
\usepackage{amssymb}
\usepackage{anysize}
\usepackage{hyperref}
\usepackage{extarrows}
\usepackage{multirow}
\usepackage{natbib}
\usepackage{stmaryrd}

\usepackage{amsmath}
\usepackage{enumitem}
\usepackage{mathtools} 

\usepackage{soul}
 
\numberwithin{equation}{section}

\theoremstyle{plain} 
\newtheorem{thm}{Theorem}[section]
\newtheorem{lem}[thm]{Lemma}

\newtheorem{pro}[thm]{Proposition}
\newtheorem{assumption}[thm]{Assumption}
\newtheorem{defn}[thm]{Definition}

\theoremstyle{remark}
\newtheorem{rem}[thm]{Remark}

\newcommand{\re}{\mathrm{Re}\,}
\newcommand{\im}{\mathrm{Im}\,}

\newcommand{\R}{{\mathbb R }}

\newcommand{\C}{{\mathbb C}}

\newcommand{\ii}{\mathrm{i}}
\newcommand{\deq}{\mathrel{\mathop:}=}

\newcommand{\dd}{\mathrm{d}}
\newcommand{\ie}{\emph{i.e., }}
\newcommand{\eg}{\emph{e.g., }}
\newcommand{\cf}{\emph{c.f., }}

\newcommand{\bs}{\boldsymbol}

\renewcommand{\mathbf}[1]{\bs{#1}}
 
\marginsize{35mm}{35mm}{38mm}{40mm}

\begin{document}

 \begin{minipage}{0.85\textwidth}
 \vspace{2.5cm}
 \end{minipage}
\begin{center}
\large\bf
On the support of the free additive convolution 
\end{center}

\renewcommand{\thefootnote}{\fnsymbol{footnote}}	
\vspace{1cm}
\begin{center}
 \begin{minipage}{0.32\textwidth}
\begin{center}
Zhigang Bao\footnotemark[1]  \\
\footnotesize {HKUST}\\
{\it mazgbao@ust.hk}
\end{center}
\end{minipage}
\begin{minipage}{0.32\textwidth}
\begin{center}
L\'aszl\'o Erd{\H o}s\footnotemark[2]  \\
\footnotesize {IST Austria}\\
{\it lerdos@ist.ac.at}
\end{center}
\end{minipage}
\begin{minipage}{0.33\textwidth}
 \begin{center}
Kevin Schnelli\footnotemark[3]\\
\footnotesize 
{KTH Royal Institute of Technology}\\
{\it schnelli@kth.se}
\end{center}
\end{minipage}
\footnotetext[1]{Supported in parts by  Hong Kong RGC Grant ECS 26301517.}
\footnotetext[2]{Supported in parts by ERC Advanced Grant RANMAT No.\ 338804.}
\footnotetext[3]{Supported in parts by  the G\"oran Gustafsson Foundation and the Swedish Research Council Grant VR-2017-05195.}

\renewcommand{\thefootnote}{\fnsymbol{footnote}}	

\end{center}

\vspace{1cm}

\begin{center}
 \begin{minipage}{0.83\textwidth}\footnotesize{
 {\bf Abstract.}  We consider the free additive convolution  of 
 two probability measures $\mu$ and~$\nu$ on the real line and show
 that $\mu\boxplus\nu$ is supported on a single interval if $\mu$ and $\nu$ each has
 single interval support.
 Moreover, the density of $\mu\boxplus\nu$ is proven to vanish as a square root near the edges of
 its support if both $\mu$ and $\nu$ have power law behavior with exponents between $-1$ and $1$
 near their edges. In particular, these results show the ubiquity of the conditions in our recent
 work on optimal local law at the spectral edges for addition of random matrices \cite{BES17}.
 
}
\end{minipage}
\end{center}

 \vspace{5mm}
 
 {\small
\footnotesize{\noindent\textit{Date}: October 27, 2018.}
 
 \vspace{2mm}

 }

\thispagestyle{headings}

\section{Introduction}

The classical convolution of two probability measures is of key interest in probability theory as it gives the law of the sum of two independent random variables. In analogy, Voiculescu~\cite{Voi86} introduced in free probability theory the free additive convolution. Let $\mu$ and $\nu$ be two Borel probability measures on the real line. Then the free additive convolution of $\mu$ and $\nu$, denoted $\mu\boxplus\nu$, is the law of $X+Y$ where $X$ and $Y$ are freely independent, self-adjoint, non-commutative random variables with laws $\mu$ and $\nu$. Though conceptually related, the classical convolution and the free convolution behave strikingly different. For example,
the classical convolution of two pure point measures is always pure point, while the free convolution always  has a non-vanishing absolutely continuous part. In particular, choosing $\mu=\nu$ as centered Bernoulli distribution, the free convolution $\mu\boxplus\mu$ is an absolutely continuous measure,
while the classical convolution is regularizing only in the sense that the $n$-fold convolution of $\mu$ becomes, upon rescaling, Gaussian in the limit of large $n$.
Note that the analogous central limit theorem for the free additive convolution yields Wigner's semicircle law in the limit.

 This basic example fittingly illustrates that, in contrast to the classical convolution, it is hard to infer, based upon intuition and heuristics, qualitative properties of the free additive convolution measure. Part of the reason is that there is no simple formula for the free additive convolution measure; it can 
 only be obtained as implicit solution to certain systems of equations. Thence the following seemingly simple question turns out quite difficult to answer: If $\mu$ and $\nu$ are both supported on a single interval, is $\mu\boxplus\nu$ then also supported on a single interval? Interestingly, while regularity properties of the free convolution measure have been extensively studied~\cite{BeV98,Bel1,Bel}, this natural problem apparently has not been studied
  in the literature. The  aim of this note is to answer this question in the affirmative for a large class of initial measures motivated by random matrix theory.
 
The emergence of the semicircle law indicates a strong link between free probability and random matrix theory. Voiculescu discovered in~\cite{Voi91} that random matrices not only provide examples of asymptotically free random variables, but can also be used to generate freeness. Conjugating symmetric matrices by independent Haar unitary matrices furnishes asymptotically free random variables. The prime example is the setup of the addition of two deterministic matrices 
in randomly chosen relative basis. In other words, if $A=A_N$ and $B = B_N$ are two sequences of deterministic Hermitian
matrices of size~$N$ and~$U$ is Haar distributed on the unitary group~$U(N)$, then~$A$ and $UBU^*$ are asymptotically free
in the limit of large~$N$ and the asymptotic eigenvalue distribution of $A + U BU^*$ is given
by the free additive convolution of the limiting eigenvalue distributions of~$A$ and~$B$. Given this convergence result, it is natural to ask about the speed of convergence and whether the convergence also holds on local scales. We have recently answered both questions  by deriving a so-called local law for the Green functions
\cf~\cite{BES17,BES16,BES15b}.
 However, the study of local laws for the model mentioned above crucially relies, in contrast to the global scale, on detailed regularity properties and the qualitative behavior of the deterministic free convolution measures.
 An objective of the present paper is to derive these decisive properties of the free convolution in the most relevant cases arising in random matrix theory.

We will focus on the free additive convolution of a class of  Jacobi type measures. These are measures
supported on a single interval with density behaving as a  power law  with exponent between -1 and 1 near the edges;
see Assumption~\ref{a.regularity of the measures} below. Wigner's semicircle law as well as the Marchenko--Pastur law are included in this class. Our main result, Theorem~\ref{le theorem 1}, asserts that the free additive convolution of two Jacobi type measures is supported on  a single interval and that its density vanishes as a square root at the two endpoints. These are the main conditions on the free convolution measure we required in our recent paper on the optimal local law at the spectral edges \cite{BES17}. 
 Theorem~\ref{le theorem 1} shows  that these assumptions  are natural.

The square root behavior at the edge is ubiquitous for densities arising in random matrix theory. The same phenomenon has been extensively
studied for  Wigner type random matrices \cite{AEK}, and more recently  for the underlying  Dyson
equation in a  general non-commutative  setup \cite{AEK18}. 
 Under some regularity condition on the matrix of variances, single interval support for the density has also been shown
\cite[Theorem 2.11]{AEK}, see also \cite[Corollary 9.4]{AEK18} for a generalization. Despite these similarities in the statements, the 
approach used for the Dyson equation is very different from the methods in the current work; simply the structures
of the underlying defining equations are not comparable.

Our proofs rely on methods from  function theory. Albeit  being introduced as an algebraic operation, the free additive convolution can be studied with complex analysis. The Stieltjes transform of the free additive convolution is related to the Stieltjes transforms of the original measures through analytic subordination. The existence of analytic subordination functions off the real line was observed in~\cite{Voi93,Bia98} and may directly be used to define the free additive convolution~\cite{BB,CG}; see Subsection~\ref{subsection: free addtive convolution} below. Function theory then provides powerful tools to study the free additive convolution and its regularization properties in great generality, that is,
 for very general Borel probability measures; see~\cite{BeV98,Bel1,Bel,BB-GG} and references therein. Specializing to Jacobi measures, we can analyze the boundary behavior of the subordination function on the real line and extract the qualitative behavior of the free convolution measure claimed in Theorem~\ref{le theorem 1}.

{\it Organization} of the paper: In Section~\ref{section main result}, we state our main results in detail, give the full definition of the free additive convolution and embed our paper in the literature. In Section~\ref{section 3}, we derive estimates on the Stieltjes transform of the free convolution and localize the subordination function on the real line. This information is then used in Section~\ref{section 4} to characterize regular edges and to prove the Theorem~\ref{le theorem 1}.

{\it Notation:} We use~$c$ and~$C$ to denote strictly positive constants. Their values may change from line to line. We denote by $\C^+$ the upper half-plane in $\C$, \ie $\C^+\deq\{z\in\C\,:\,\im z>0\}$.

{\it Acknowledgment:} We thank Serban Belinschi who provided insight and expertise related to the behavior of the subordination functions. We thank the anonymous referee for comments and help with the literature.

\section{Main results}\label{section main result}

Let $\mu_\alpha$ and $\mu_\beta$ be two Borel probability measures on $\R$. In this paper we 
  study support and regularity properties  of the free additive convolution measure, $\mu_\alpha\boxplus\mu_\beta$; 
  see Subsection~\ref{subsection: free addtive convolution} for the precise  definition of the free additive convolution. We focus on the case when  $\mu_\alpha$ and~$\mu_\beta$ are both absolutely continuous and have single interval support. Moreover, we assume that they are of {\it Jacobi} type 
  by which we mean that they vanish as  a power-law  at the edges of the support. More precisely, we will assume that $\mu_\alpha$ and $\mu_\beta$ satisfy the following.
\begin{assumption} \label{a.regularity of the measures} The measures $\mu_\alpha$ and $\mu_\beta$ are compactly supported and
centered probability measures that are absolutely continuous  with respect to Lebesgue measure with
 density functions $\rho_\alpha$ and $\rho_\beta$ such that:
\begin{itemize}
\item[$(i)$] Both density functions $\rho_\alpha$ and $\rho_\beta$ have single non-empty interval supports, $[E_-^\alpha,E_+^\alpha]$ and $[E_-^\beta,E_+^\beta]$, respectively.

\item[$(ii)$]  The density functions have a power law behavior: there are exponents $ -1<t_-^\alpha, t_+^\alpha<1$ and~ $-1<t_-^\beta,t_+^\beta<1$ such that 
\begin{align}\label{le constant C in measures}
C^{-1}\leq\frac{\rho_\alpha(x)}{(x-E_-^\alpha)^{t_-^\alpha}(E_+^\alpha-x)^{t_+^\alpha}}\leq C\,,\qquad \textrm{for a.e.}\quad	 x \in [E_-^\alpha, E_+^\alpha]\,,  \nonumber\\
C^{-1}\leq\frac{\rho_\beta(x)}{(x-E_-^\beta)^{t_-^\beta}(E_+^\beta-x)^{t_+^\beta}}\leq C\,,\qquad \textrm{for a.e.}\quad	 x \in [E_-^\beta, E_+^\beta]\,,
\end{align}
hold for some  positive constant $C>1$. 
\end{itemize} 
\end{assumption}

We will explain in Section~\ref{subsection: free addtive convolution} that under these conditions the free convolution measure $\mu_\alpha\boxplus \mu_\beta$
is known to have a continuous and bounded density $\rho$. 
Our following main result  shows that $\rho$ is supported on a single interval with square root singularities at the edges.

\begin{thm}\label{le theorem 1}
 Assume that $\mu_\alpha$ and $\mu_\beta$ satisfy Assumption~\ref{a.regularity of the measures}, in particular these measures have single interval support. 
 Then their free additive convolution $\mu_\alpha\boxplus \mu_\beta$ is also supported on a single compact interval that we denote by $[E_-, E_+]$. Moreover 
 $E_-<0$, $E_+>0$ and there exists $C>1$ such that	
 \begin{align}\label{eq theorem 1}
  C^{-1}\le\frac{\rho(x)}{\sqrt{x-E_-}\sqrt{E_+-x}} \le C\,,\qquad\qquad  \forall x\in[E_-,E_+]\,,
 \end{align}
where $\rho$ denotes the continuous  density function of $\mu_\alpha\boxplus\mu_\beta$.
\end{thm}

\begin{rem}\label{rem:parameters}  It can be checked from our proofs that the constant $C$ in~\eqref{eq theorem 1} depends 
only  on certain {\it control parameters}, namely on the constant in~\eqref{le constant C in measures}, on the exponents in~\eqref{le constant C in measures}, the second moments of $\mu_\alpha$ and $\mu_\beta$, and on
 the constant $c>0$ serving as a lower bound in
\begin{align}\label{the minis}
  |m_{\mu_\alpha}(E)| \ge c \,,\qquad 
  |m_{\mu_\beta}(E')|\ge c 
\end{align}
for all $E\in [ E_-^\alpha-1,E_-^\alpha]\cup[E_+^\alpha,E_+^\alpha+1]$ and for all $E'\in[E_-^\beta-1, E_-^\beta]\cup[E_+^\beta,E_+^\beta+1]$. 

 Notice that the Stieltjes transforms of $\mu_\alpha$ and $\mu_\beta$ can be extended  as non-tangential limits to the real axis Lebesgue-almost everywhere. Outside the support of the measure~$\mu_\alpha$, respectively $\mu_\beta$, these extensions are real valued and analytic. Since $\mu_\alpha$ and $\mu_\beta$ have single interval supports, the Stieltjes transforms cannot have any zeros outside their supports. In particular, there is indeed a positive
lower bound $c$ in \eqref{the minis} on the indicated intervals.

In general, we apply  a similar convention throughout paper: when 
we  state that some constant  depends on the two input measures $\mu_\alpha$ and $\mu_\beta$, we mean that
it depends only on the above  control parameters.

\end{rem}

\begin{rem}
The assumption that the exponents  $t_-^\alpha, t_+^\alpha,t_-^\beta$ and $t_+^\beta$ in \eqref{le constant C in measures} are bigger than $-1$ is 
 necessary to have finite measures. 
In general, the assumption  that they are smaller than $1$ is necessary to have a square root behavior at the edges of the free convolution. Indeed, 
if one of the exponents exceeds $1$, it can happen that an edge behavior other than the square root emerges; see~\cite{LS13,LS16bis} for a detailed analysis of 
a special case  when one of the measures is a the semicircle law and the other has a convex behavior at the endpoints of the supports. However, we still expect that the free additive convolution of two Jacobi
 measures with general exponents is supported on a single interval. We point out that in most applications
 the endpoint exponents are strictly below 1, so our theorem applies.
\end{rem}

\subsection{Free additive convolution}\label{subsection: free addtive convolution}
In this subsection we review the definition of the free additive convolution in detail. We start with the Stieltjes transform:
For any probability measure $\mu$ on $\R$, its Stieltjes transform is defined as 
\begin{align*}
m_\mu(z)\deq\int_{\mathbb{R}} \frac{1}{x-z}\, \dd\mu(x)\,, \qquad \qquad z\in \mathbb{C}^+\,.
\end{align*}
We further denote by $F_\mu$ the negative reciprocal Stieltjes transform of $\mu$, i.e.,
\begin{align*}
F_\mu(z)\deq-\frac{1}{m_\mu(z)}\,.
\end{align*}
Note that $F_{\mu}\,:\C^+\rightarrow\C^+$ is analytic and satisfies
\begin{align*}
 \lim_{\eta\nearrow\infty}\frac{F_{\mu}(\ii\eta)}{\ii\eta}=1\,.
\end{align*}

Conversely, if $F\,:\,\C^+\rightarrow\C^+$ is an analytic function with
$\lim_{\eta\nearrow\infty} F(\ii\eta)/\ii\eta=1$, then $F$ is the negative reciprocal Stieltjes transform of a probability
measure $\mu$, \ie $F(z) = F_{\mu}(z)$, for all $z\in\C^+$; see \eg~\cite{Ak}.

Voiculescu introduced the free additive convolution of Borel probability measures on $\R$ in the groundbreaking paper~\cite{Voi86} in an algebraic setup as the distribution of the sum of two freely independent non-commutative random variables. Our starting point is the following result which can be used to define the free additive convolution in an analytic setup.

\begin{pro}[Theorem 4.1 in~\cite{BB}, Theorem~2.1 in~\cite{CG}]\label{le prop 1}
Given two Borel probability measures, $\mu_\alpha$ and $\mu_\beta$, on $\R$, there exist unique analytic functions, $\omega_\alpha,\omega_\beta\,:\,\C^+\rightarrow \C^+$, such that,
 \begin{itemize}[noitemsep,topsep=0pt,partopsep=0pt,parsep=0pt]
  \item[$(i)$] for all $z\in \C^+$, $\im \omega_\alpha(z),\,\im \omega_\beta(z)\ge \im z$, and
  \begin{align}\label{le limit of omega}
  \lim_{\eta\nearrow\infty}\frac{\omega_\alpha(\ii\eta)}{\ii\eta}=\lim_{\eta\nearrow\infty}\frac{\omega_\beta(\ii\eta)}{\ii\eta}=1\,;
  \end{align}
  \item[$(ii)$] for all $z\in\C^+$, 
  \begin{align}\label{le definiting equations}
   F_{\mu_\alpha}(\omega_{\beta}(z))=F_{\mu_\beta}(\omega_{\alpha}(z))\,,\qquad\qquad \omega_\alpha(z)+\omega_\beta(z)-z=F_{\mu_\alpha}(\omega_{\beta}(z))\,.
  \end{align}
 \end{itemize}
\end{pro}

The analytic function $F\,:\,\C^+\rightarrow \C^+$ defined by
\begin{align}\label{le kkv}
 F(z)\deq F_{\mu_\alpha}(\omega_{\beta}(z))=F_{\mu_\beta}(\omega_{\alpha}(z))\,,
\end{align}
 is by part $(i)$ of Proposition~\ref{le prop 1} the negative reciprocal Stieltjes transform of a probability measure $\mu$, called the free additive convolution of $\mu_\alpha$ and $\mu_\beta$ and denoted by $\mu\equiv\mu_\alpha\boxplus\mu_\beta$. The functions $\omega_\alpha$ and $\omega_\beta$ are referred to as the {\it subordination functions}. The subordination phenomenon
was first noted by Voiculescu~\cite{Voi93} in a generic situation
and in full generality by Biane~\cite{Bia98}.

Choosing $\mu_\alpha$ arbitrary and $\mu_\beta$ as delta mass at $x\in \R$, it is easy to check that $\mu_\alpha\boxplus \mu_\beta  $ simply is $\mu_\alpha$ shifted by $x$. We therefore exclude this trivial case from our considerations. Moreover, by a simple shift we may without lost of generality assume that $\mu_\alpha$ and $\mu_\beta$ are centered measures; see~Assumption~\ref{a.regularity of the measures}.

The atoms of $\mu_\alpha\boxplus\mu_\beta$ are determined as follows.  A point $w\in\R$ is an atom of $\mu_\alpha\boxplus\mu_\beta$ if and only if there exist $x,y\in\R$ such that $w=x+y$ and $\mu_\alpha(\{x\})+\mu_\beta(\{y\})>1$; see [Theorem~7.4,~\cite{BeV98}]. Thus in particular under Assumption~\ref{a.regularity of the measures}, the free additive convolution does not have any atoms. The boundary behavior of the subordination functions~$\omega_\alpha$ and~$\omega_\beta$ was studied by Belinschi in a series of papers~\cite{Bel1,Bel,Bel2} where he proved the following two results. For sake of simplicity, we limit the discussion to compactly supported measures.

\begin{pro}[Theorem~3.3 in~\cite{Bel}, Theorem~6 in~\cite{Bel2}]\label{prop extension}
 Let $\mu_\alpha$ and $\mu_\beta$ be compactly supported Borel probability measures on $\R$, none of them being a single point mass. Then the subordination functions  $\omega_\alpha$, $\omega_\beta\,:\, \C^+\to\C^+$ extend continuously to $\C^+\cup\R$ as functions with values in $\C^+\cup \R\cup \{\infty\}$.
 \end{pro}
 
In Theorem~4.1 of~\cite{Bel}, Belinschi proved that the singular continuous part of $\mu_\alpha\boxplus\mu_\beta$ is always zero, and that the absolutely continuous part of $\mu_\alpha\boxplus\mu_\beta$ is always nonzero and admits a continuous density function. We denote this density function by $\rho$. Summing up, we have under Assumption~\ref{a.regularity of the measures} the following regularity result:
\begin{lem}
 Let $\mu_\alpha$ and $\mu_\beta$ satisfy Assumption~\ref{a.regularity of the measures}. Then the free additive convolution measure $\mu_\alpha\boxplus\mu_\beta$ is absolutely continuous with respect to Lebesgue measure and admits a continuous and bounded density function $\rho$ that is real analytic wherever strictly positive.
\end{lem}

\subsection{Previous results}
We start by results concerning the supports of free convolution measures. Biane studied in~\cite{B} the free convolution of the semicircle law with an arbitrary probability measure: Denote by $\sigma_t(\dd x)=\frac{1}{2\pi t}\sqrt{4t-x^2}\,1_{[-2\sqrt{t},2\sqrt{t}]}(x)\dd x$ the density of the semicircle law of variance $t$. Given a probability measure $\lambda$ on $\R$, we obtain a one-parameter family of probability measures, the so-called semi-circular flow, by setting $\mu_t\deq \lambda\boxplus\sigma_t$, $t>0$. Biane proved that the number of connected components of $\mu_t$ is a non-increasing function of $t$ and that the continuous density of $\mu_t$ satisfies
\begin{align}
 |\mu_t(x)|\le \Big(\frac{3}{4\pi^3t^2|x-x_0|}\Big)^{1/3}\,,\qquad\qquad x\in\R\,,
\end{align}
where $x_0$ is the closed point to $x$ in the complement of the interior of the support of $\mu_t$. For further results of the semi-circular flow we refer to~\cite{ShchT11}. 

In the appendix to~\cite{BPB99}, Biane obtained support properties of freely stable laws. 

Voiculescu proved in~\cite{Voi85} the free central limit theorem for the addition of freely independent non-commutative random variables in terms of convergence of moments. The convergence to the limiting semicircle distribution turned out to be much stronger: already after a finite number of free convolutions, the distribution of the finite free sum becomes absolutely continuous with respect to Lebesgue measure. This so-called superconvergence was established first by Bercovici and Voiculescu in~\cite{BeV95} and subsequently refined by Kargin~\cite{Kargin2007} and Wang~\cite{Wan10}.

 The $n$-fold free convolution power $\lambda^{\boxplus n}$ of a probability measure $\lambda$ on $\R$ can be embedded in a one-parameter family $\{\lambda^{\boxplus t},\,t\ge 1\}$ with the semigroup structure $\lambda^{\boxplus t_1}\boxplus\lambda^{\boxplus t_2}=\lambda^{\boxplus(t_1+t_2)}$, $t_1,t_2\ge 1$;~see~\cite{BeV95,NiSp96}. Huang proved in~\cite{Hua14} that the support of $\lambda^{\boxplus t}$, $t>1$, consists of at most finitely many atoms and countably many intervals and that the number of the components of the support of $\lambda^{\boxplus t}$ is a decreasing function of $t$. We mention that the system of subordination equations in both cases, the semicircular flow and the free convolution semigroup, reduce to a single equation rendering the support analysis much simpler.

The results in~\cite{BeV95,B,Hua14,Kargin2012,Wan10} seemingly suggest that convolving freely reduces the number of connected components in the support. Yet superconvergence results to freely stable laws in~\cite{BWZ17} and results in~\cite{BeW2008} might betoken that the situation is in general not quite as clear.

For the free addition of two Jacobi measures, Olver and Rao proved in~\cite{ON} that if $\mu_\alpha$ is a Jacobi measure with $t_\pm^\alpha=1/2$ and $\mu_\beta$ is a Jacobi measure whose Stieltjes transform is single-valued, then $\mu_\alpha\boxplus\mu_\beta$ is a Jacobi measure which exhibits a square root behavior at its edges.

In~\cite{BES17}, Section 3, we studied the behavior of free additive convolution at the smallest and largest endpoints of its support. We showed that under similar conditions  to the current Assumption~\ref{a.regularity of the measures} and the additional assumption that at least one of the  following two bounds 
\begin{equation}\label{mbound}
  \sup_{z\in \C^+}|m_{\mu_\alpha}(z)|\le C\,, \qquad\qquad   \sup_{z\in \C^+}|m_{\mu_\beta}(z)|\le C\,,
\end{equation}
holds, for some positive constant $C$, that $\mu_\alpha\boxplus\mu_\beta$ vanishes as a square root at the smallest and largest endpoint of its support. Theorem~\ref{le theorem 1} overpasses these results by removing the unnatural assumption in~\eqref{mbound}, but more importantly, it asserts that the free additive convolution, under Assumption~\ref{a.regularity of the measures}  has only two edges, \ie $\mu_\alpha\boxplus\mu_\beta$ is supported on a single interval, and its density is strictly positive inside the support.

Finally, we mention that the linear stability of the system~\eqref{le definiting equations} was first effectively studied by Kargin in~\cite{Kargin2012} under some genericity conditions. In~\cite{BES15}, we showed that these conditions are fulfilled in the regular bulk and in~\cite{BES17} we extended the stability results to square root edges where the system is only quadratically stable.

\section{Properties of the Stieltjes transform and the subordination functions}\label{section 3}

We will repeatedly use the following integral representation for Pick functions; see \eg Chapter III of~\cite{Ak} for a reference.

\begin{lem}\label{lemma pick}
 Let $f\,:\,\C^+\rightarrow \C^+\cup\R$ be an analytic function. Then there exists a positive Borel measure $\mu$ on $\R$ and $a\in\R$, $b\ge 0$ such that
 \begin{align}\label{le pick}
  f(\omega)=a+b\,\omega+\int_\R\left(\frac{1}{x-\omega}-\frac{x}{1+x^2}\right)\dd\mu(x)\,,\qquad\quad \omega\in\C^+\,, 
 \end{align}
and 
\begin{align}
 \int_\R\frac{1}{1+x^2}\,\dd\mu(x)<\infty\,.
\end{align}

\end{lem}

The negative reciprocal Stieltjes transforms of $\mu_\alpha$ and $\mu_\beta$ enjoy the following properties.	

\begin{lem}\label{lemm on hat measures}
Let $\mu_\alpha$ and $\mu_\beta$ satisfy Assumption~\ref{a.regularity of the measures}. Then there exist Borel measures $\widehat\mu_\alpha$ and $\widehat\mu_\beta$ such that
\begin{align}\label{it is Stieltjes transform}
 F_{\mu_\alpha}(\omega)-\omega&=\int_\R\frac{1}{x-\omega}\,\dd\widehat\mu_\alpha(x)\,, \qquad\qquad \omega\in \C^+\,,
\end{align}
with
\begin{align}\label{finit measure}
 0<\widehat\mu_\alpha(\R)=\int_\R x^2\dd\mu_\alpha(x)<\infty
\end{align}
and
\begin{align}\label{same support}
 \mathrm{supp}\,\widehat\mu_\alpha=\mathrm{supp}\,\mu_\alpha\,.
\end{align}
In particular, $\widehat\mu_\alpha$ is a finite compactly supported Borel measure. The same statements hold true when the index $\alpha$ is replaced by $\beta$.
\end{lem}
\begin{rem}
 Equations~\eqref{it is Stieltjes transform} and~\eqref{finit measure} are well-known result, see e.g.\ ~Proposition~2.2 in~\cite{Maa}.  For convenience we include their proofs  below.
\end{rem}

\begin{proof}[Proof of Lemma~\ref{lemm on hat measures}]
 We start by noticing that the Stieltjes transform of $\mu_\alpha$ admits the following asymptotic expansion,
 \begin{align*}
  m_{\mu_\alpha}(\ii\eta)&=\frac{1}{-\ii \eta}\int_\R\dd\mu_\alpha(x)-\frac{1}{(\ii\eta)^2}\int_\R x\,\dd\mu_\alpha(x)-\frac{1}{(\ii\eta)^3}\int_\R x^2\,\dd\mu_\alpha(x)+O(\eta^{-4})\,,\\
  &=\frac{1}{-\ii\eta}+\frac{1}{\ii \eta^3}\int_\R x^2\,\dd\mu_\alpha(x)+O(\eta^{-4})\,,
 \end{align*}
as $\eta\nearrow\infty$, where we used that $\mu_\alpha$ is a centered probability measure. Hence, taking the negative reciprocal, we find in the limit $\eta\nearrow\infty$ that
\begin{align}\label{used to compare}
 F_{\mu_\alpha}(\ii\eta)-\ii\eta=-\frac{1}{\ii\eta}\int_\R x^2\dd\mu_\alpha(x)+O(\eta^{-2})\,.
\end{align}

Next, note that $f(\omega)\deq F_{\mu_\alpha}(\omega) -\omega$ satisfies
\begin{align*}
 \im f(\omega)=\frac{\im m_{\mu_\alpha}(\omega)-\im \omega|m_{\mu_\alpha}(\omega)|^2}{|m_{\mu_\alpha}(\omega)|^2}=\im \omega\bigg(\int_\R\frac{\dd\mu_\alpha(x)}{|x-\omega|^2}\bigg/\left|\int_\R\frac{\dd\mu_\alpha(x)}{x-\omega}\right|^2-1\bigg)>0\,,
\end{align*}
for $\omega\in\C^+$, where we used that $\mu_\alpha$ is supported at more than one point by Assumption~\ref{a.regularity of the measures}. Thus Lemma~\ref{lemma pick}  applies to $f$ with some measure $\mu=:\widehat\mu_\alpha$. Comparing~\eqref{used to compare} with~\eqref{le pick}, we conclude that $b=0$ and thus
 \begin{align}\label{FF}
F_{\mu_\alpha}(\omega) -\omega=f(\omega) = a+\int_\R\left(\frac{1}{x-\omega}-\frac{x}{1+x^2}\right)\dd\widehat\mu_\alpha(x)\,.
 \end{align}
Choosing $\omega=\ii\eta$ in~\eqref{FF} and taking $\eta\nearrow\infty$, comparison with~\eqref{used to compare} immediately yields
that $a=\int_\R\frac{x}{1+x^2}\dd\widehat\mu_\alpha$  and thus~\eqref{it is Stieltjes transform} holds. Furthermore,
\eqref{finit measure} also holds by comparing the coefficient of $\eta^{-1}$
on the right sides of~\eqref{it is Stieltjes transform} and~\eqref{used to compare}  in the large $\eta$ limit.

 Taking the imaginary parts in~\eqref{FF} we further obtain
\begin{align}\label{le smanl1}
\im f(\omega)=\frac{\im m_{\mu_\alpha}(\omega)}{|m_{\mu_\alpha}(\omega)|^2}-\im \omega= \int_\R\frac{\im\omega}{|x-\omega|^2} \;\dd\widehat\mu_\alpha(x)
 \,,\qquad\qquad \omega\in\C^+\,.
\end{align}
From Assumption~\ref{a.regularity of the measures}, we know that the extension of $m_{\mu_\alpha}$ to $\R$ is continuous and real valued outside the support of $\mu_\alpha$. Since in addition~$\mu_\alpha$ is supported on a single interval, we conclude that $m_{\mu_\alpha}$ does not have any zeros on $\R\backslash\mathrm{supp}\,\mu_\alpha$. Hence taking the limit $\im\omega\searrow 0$ in~\eqref{le smanl1},
we conclude by the Stieltjes inversion formula that $\widehat\mu_\alpha$ is absolutely continuous with respect to Lebesgue measure on the complement of $\mathrm{supp}\,\mu_\alpha$ with vanishing density function. Hence we must have $\mathrm{supp}\,\widehat\mu_\alpha \subseteq \mathrm{supp}\,\mu_\alpha$. In particular $\widehat\mu_\alpha$ is a finite compactly supported measure. 

Finally, to conclude~\eqref{same support}, 
 we need to prove the opposite containment, \ie that $\mathrm{supp}\,\widehat\mu_\alpha  \supseteq  \mathrm{supp}\,\mu_\alpha$. Suppose, on the contrary, that $\mathrm{supp}\,\widehat\mu_\alpha$
is a proper subset of  $\mathrm{supp}\,\mu_\alpha=[E_-^\alpha, E_+^\alpha]$. Then we can find a non-empty open interval $I\subset\mathrm{supp}\,\mu_\alpha\backslash\mathrm{supp}\,\widehat{\mu}_\alpha$ such that $f(\omega)\,:\,\C^+\rightarrow\C^+$ extends continuously to $I$ with $\im f(\omega)=0$, for all $\omega\in I$. Then by the Schwarz reflection principle, $f$ extends analytically through $I$ and, hence, $m$ is meromorphic on $I$. However, since $I\subset\mathrm{supp}\,\mu_\alpha$, we have $\lim_{\eta\searrow 0}\im m_{\mu_\alpha}(\omega+\ii\eta)>0$ by Assumption~\ref{a.regularity of the measures}, for almost all $\omega\in I$. Since $m(\omega)$ is meromorphic on $I$ and $\im f(\omega)=\im m(\omega)/|m(\omega)|^2$, $\omega\in I$, we hence also have $\lim_{\eta\searrow 0}\im f(\omega+\ii\eta)>0$ for almost all $\omega\in I$, a contradiction to $\im f(\omega)=0$, for all $\omega\in  I$. Thus $I$ must be empty and we have $\mathrm{supp}\,\widehat\mu_\alpha = \mathrm{supp}\,\mu_\alpha$. This proves~\eqref{same support} and concludes the proof of Lemma~\ref{lemm on hat measures}. 
\end{proof}

\begin{rem}
 As the measures $\widehat\mu_\alpha$ and $\widehat\mu_\beta$ are finite and compactly supported, we have by dominated convergence that
 \begin{align}\label{we can do what we want}
  F'_{\mu_\alpha}(\omega)-1=\int_\R\frac{1}{(x-\omega)^2}\,\dd\widehat\mu_\alpha(x)\,,\qquad F''_{\mu_\alpha}(\omega)=\int_\R\frac{1}{(x-\omega)^3}\,\dd\widehat\mu_\alpha(x)\,,
 \end{align}
for all $\omega\in\C^+\cup\R\backslash\,\mathrm{supp}\,\widehat\mu_\alpha$; the same relations hold with the $\alpha$ changed to $\beta$.
\end{rem}

\subsection{Bounds on the subordination functions} The goal of this subsection is to control the imaginary parts of the subordination functions within a sufficiently large neighborhood of the support of the free convolution measure.

We first introduce the domain of the spectral parameter $z$ we will be working on. Let $\mathcal{J}\subset\R$ be the interval
\begin{align}\label{le J}
 \mathcal{J}\deq\{E\in\R\,:\, E_-^\alpha+E_-^\beta-1\le E\le E_+^\alpha+E_+^\beta+1\,\}\,.
\end{align}
Then we introduce the domain
\begin{align}\label{le domain}
 \mathcal{E}\deq\{z=E+\ii\eta\in\C^+\cup\R\,: E\in\mathcal{J}\,, 0\le\eta\le 1\}\,.
\end{align}
 Lemma~3.1 of~\cite{Voi86} shows that $\mathrm{supp}\,\mu_\alpha\boxplus\mu_\beta\subset \mathcal{J}$ and we can therefore restrict the discussion to that interval, respectively to $\mathcal{E}$. Finally, we use the shorthand
$$m(z)= m_{\mu_\alpha\boxplus\mu_\beta}(z)\,,\qquad\qquad z\in\C^+\,,$$ to denote the Stieltjes transform of $\mu_\alpha\boxplus\mu_\beta$, and also its continuous extension to $\mathcal{E}$.

The following is the main result of this subsection.
\begin{pro}\label{le pro links}
Assume that $\mu_\alpha$ and $\mu_\beta$ satisfy Assumption~\ref{a.regularity of the measures}. Then there is a constant $C\ge 1$, depending on $\mu_\alpha$ and $\mu_\beta$ via their control parameters such that
\begin{align}\label{le links}
 C^{-1}\im m(z)\le \im \omega_\alpha(z)\le C\,\im m(z)\,,\nonumber\\
 \qquad  C^{-1}\im m(z)\le \im \omega_\beta(z)\le C\,\im m(z)\,,
\end{align}
for all $z\in\mathcal{E}$. 
\end{pro}

We split the proof of Proposition~\ref{le pro links} in several steps.
We start with two definitions.
\begin{defn} ($I_\alpha$, $I_\beta$) Define the functions $I_\alpha$ and $I_\beta$ by setting
\begin{align}\label{le I ohne hat}
I_\alpha(\omega)\deq\int_\R\frac{\dd\mu_\alpha(x)}{|x-\omega|^2}\,,\qquad I_\beta(\omega)\deq\int_\R\frac{\dd\mu_\beta(x)}{|x-\omega|^2}\,,\qquad\qquad \omega\in\C^+\,.
\end{align}
\end{defn}
\begin{rem}
Note that $I_\alpha$ and $I_\beta$ extend continuously to the real line outside the respective supports of $\mu_\alpha$ or $\mu_\beta$. Moreover, from~\eqref{le kkv} we note that
\begin{align}\label{le to be noted}
 \im m(z)=\im \omega_\beta(z)\cdot I_\alpha(\omega_\beta(z))=\im \omega_\alpha(z)\cdot I_\beta(\omega_\alpha(z))\,,\qquad\quad z\in\C^+\,.
\end{align}
\end{rem}

\begin{defn} ($\widehat{I}_\alpha$, $\widehat{I}_\beta$)
Let $\widehat\mu_\alpha$ and $\widehat\mu_\beta$ denote the measures from Lemma~\ref{lemm on hat measures}. Set then
\begin{align}\label{le I hat}
 \widehat I_\alpha(\omega)\deq\int_\R\frac{\dd\widehat\mu_\alpha(x)}{|x-\omega|^2}\,,\qquad\qquad \widehat I_\beta(\omega)\deq\int_\R\frac{\dd\widehat\mu_\beta(x)}{|x-\omega|^2}\,,\qquad\qquad  \omega\in\C^+\,.
\end{align}
\end{defn}
\begin{rem}
 Taking the imaginary part in~\eqref{it is Stieltjes transform} and using~\eqref{le definiting equations}, we find from~\eqref{le I hat} that
 \begin{align}\label{le super trup}
  \widehat I_\alpha(\omega_\beta(z))=\frac{\im \omega_\alpha(z)-\im z}{\im \omega_\beta(z)}\,,\quad\widehat I_\beta(\omega_\alpha(z))=\frac{\im \omega_\beta(z)-\im z}{\im \omega_\alpha(z)}\,,\qquad z\in\C^+\,.
 \end{align}
Hence, since $\im\omega_\alpha(z)\ge \im z$, $\im\omega_\beta(z)\ge \im z$ by Proposition~\ref{le prop 1}, we further find that
 \begin{align}\label{smaller than one}
  \widehat I_\alpha(\omega_\beta(z))\cdot\widehat I_{\beta}(\omega_\alpha(z))\le 1\,,\qquad\qquad z\in\C^+\,.
 \end{align}
\end{rem}

 The following result shows that the subordination functions are uniformly bounded, under Assumption~\ref{a.regularity of the measures}, on $\mathcal{E}$.
\begin{lem}[Lemma 3.2.~of~\cite{BES17}]\label{lemma bound on the subi}
Assume that $\mu_\alpha$ and $\mu_\beta$ satisfy Assumption~\ref{a.regularity of the measures}. Then
\begin{align}\label{upper bound on omegas}
 |\omega_\alpha(z)|\le C\,,\qquad |\omega_\beta(z)|\le C\,,
\end{align}
uniformly in $z\in\mathcal{E}$, with constants depending on $\mu_\alpha$ and $\mu_\beta$ via their control parameters.
\end{lem}

From~\eqref{upper bound on omegas} we directly get the following estimates.
\begin{lem}\label{lemma x}
 Assume that $\mu_\alpha$ and $\mu_\beta$ satisfy Assumption~\ref{a.regularity of the measures}. Let $\mathcal{E}$ be the domain defined in~\eqref{le domain}. Then we have
 \begin{align}\label{det forsta}
  \inf_{z\in\mathcal{E}} I_\alpha(\omega_\beta(z))\ge c\,,\qquad\qquad \inf_{z\in\mathcal{E}}I_\beta(\omega_\alpha(z))\ge c\,,
 \end{align}
for some constant $c>0$. Similarly, we have
\begin{align}\label{det andra}
  \inf_{z\in\mathcal{E}} \widehat I_\alpha(\omega_\beta(z))\ge c'\,,\qquad\qquad \inf_{z\in\mathcal{E}}\widehat I_\beta(\omega_\alpha(z))\ge c'\,,
\end{align}
for some constant $c'>0$. In particular, we have
\begin{align}\label{det tredje}
 \sup_{z\in\mathcal{E}}\widehat I_\alpha(\omega_\beta(z))\le C'\,, \qquad\qquad\sup_{z\in\mathcal{E}}\widehat I_\beta(\omega_\alpha(z))\le C'\,,
\end{align}
for some constant $C'$.
\end{lem}
\begin{proof}
 The estimates in~\eqref{det forsta} follow directly from the definitions of $I_\alpha$, $I_\beta$ in~\eqref{le I ohne hat} and the upper bounds in~\eqref{upper bound on omegas}. The estimates in~\eqref{det andra} follow from the definitions of $\widehat I_\alpha$, $\widehat I_\beta$ in~\eqref{le I hat},~\eqref{finit measure} and the upper bounds in~\eqref{upper bound on omegas}. Finally,~\eqref{det tredje} follows by combining~\eqref{det andra} and~\eqref{smaller than one}.
\end{proof}
The estimates in Lemma~\ref{lemma x} are complemented by the following result. Define
\begin{align}\label{le distance functions}
 d_\alpha(\omega)\deq\mathrm{dist}\,(\omega,\mathrm{supp}\,\mu_\alpha)\,,\qquad d_\beta(\omega)\deq\mathrm{dist}\,(\omega,\mathrm{supp}\,\mu_\beta)\,.
\end{align}

\begin{lem}\label{gap corollary}
 Under the assumptions of Lemma~\ref{lemma x}, we have the bounds
 \begin{align}\label{le gap}
  \inf_{z\in\mathcal{E}}d_\alpha(\omega_\beta(z))\ge g\,,\qquad\qquad \inf_{z\in\mathcal{E}}d_\beta(\omega_\alpha(z))\ge g\,,
 \end{align}
for a strictly positive constant $g$ depending only on $\mu_\alpha$ and $\mu_\beta$. Moreover, we have that
\begin{align}\label{gap bounds everything}
 \sup_{z\in\mathcal{E}} I_\alpha(\omega_\beta(z))\le \frac{1}{g^2}\,,\qquad\qquad \sup_{z\in\mathcal{E}} I_\beta(\omega_\alpha(z))\le \frac{1}{g^2}\,.
\end{align}

\end{lem}
\begin{proof}

Recall the definition of $\widehat I_\alpha$ in~\eqref{le I hat}. Using~\eqref{le smanl1} we can write
\begin{align}\label{le IKC}
 \widehat I_\alpha (\omega) = \frac{\im m_{\mu_\alpha}(\omega)}{|m_{\mu_\alpha}(\omega)|^2 \im \omega}-1= \frac{ \int_\R\frac{\dd \mu_\alpha(x)}{|x-\omega|^2} }{ \Big| \int_\R \frac{\dd \mu_\alpha(x)}{x-\omega}\Big|^2}-1\,.
\end{align}
We now claim that $\widehat I_\alpha (\omega)\nearrow\infty$ as $\omega$ approaches $\mathrm{supp}\,\mu_\alpha$ in $\mathcal{D}_\alpha\deq\C^+\cup\R\backslash{\mathrm{supp}\,\mu_\alpha}$. To do so we distinguish two cases: we first study $\widehat I_\alpha(\omega)$ with $\omega$ in a neighborhood of the edges $E_-^\alpha$ respectively $E_+^\alpha$, and then study $\widehat I_\alpha(\omega)$ for $\omega$ inside the bulk separately.

An elementary computation shows that, for $\omega\in\C^+$ satisfying $|\omega-E_-^\alpha|\le \delta$ with some (small) $\delta>0$, 
\begin{align}\label{from the computational lemma}
\int_{\R}\frac{\dd\mu_\alpha(x)}{|x-\omega|^2}\ge c\begin{cases}
                                                                                  \frac{(\re\omega-E_-^\alpha)^{t_-^\alpha}}{\im \omega}\,, &\textrm{ if }\qquad \re\omega-E_-^\alpha>\im \omega\,,\\
                                                                                   (E_-^\alpha-\re\omega)^{t_-^\alpha-1}\,, &\textrm{ if }\qquad \re\omega-E_-^\alpha< -\im\omega\,,\\
                                                                                    (\im \omega)^{t_-^\alpha-1}\,, &\textrm{ if }\qquad \im\omega\ge|\re\omega-E_-^\alpha|\,,
                                                                                   \end{cases}
\end{align}
 for some constant $c>0$ depending on $\delta$; see~\eg Lemma~3.4 in~\cite{BES17}.

For $t_-^\alpha\ge0$, we have for $\omega\in \C^+$ satisfying $|\omega-E_-^\alpha|\le \delta$, the estimate
\begin{align}\label{le diverging2}
    \left|\int_\R \frac{\dd \mu_\alpha(x)}{x-\omega}\right| \le C\begin{cases}
                                                                                  |\log \im \omega|\,, &\textrm{ if }\qquad \re\omega-E_-^\alpha>\im \omega\,,\\
                                                                                   |\log (E_-^\alpha-\re\omega)|\,, &\textrm{ if }\qquad \re\omega-E_-^\alpha< -\im\omega\,,\\
                                                                                    |\log\im \omega|\,, &\textrm{ if }\qquad \im\omega\ge|\re\omega-E_-^\alpha|\,,
                                                                                   \end{cases}
\end{align}
for some constant $C$ depending on $\delta$. For $t_-^\alpha<0$, we have,  for $\omega\in \C^+$ satisfying $|\omega-E_-^\alpha|\le \delta$,
\begin{align}\label{le diverging3}
    \left|\int_\R \frac{\dd \mu_\alpha(x)}{x-\omega}\right| \le C\begin{cases}
                                                |\log \im\omega|\,(\re\omega-E_-^\alpha)^{t_-^\alpha}\,, &\textrm{ if }\qquad \re\omega-E_-^\alpha>\im \omega\,,\\
                                                 |\log  (E_-^\alpha-\re\omega)|\,                                  |\re\omega-E_-^\alpha|^{t_-^\alpha}\,, &\textrm{ if }\qquad \re\omega-E_-^\alpha< -\im\omega\,,\\
                                                                             |\log \im\omega|\,        (\im \omega)^{t_-^\alpha}\,, &\textrm{ if }\qquad \im\omega\ge|\re\omega-E_-^\alpha|\,,
                                                        \end{cases}\end{align}
for some strictly positive constants $C$ depending on $\delta$.

Next, set
\begin{align}\label{mr t}
 T(\omega)\deq \frac{ \int_\R\frac{\dd \mu_\alpha(x)}{|x-\omega|^2} }{ \Big| \int_\R \frac{\dd \mu_\alpha(x)}{x-\omega}\Big|^2}\,,
\end{align}
with $\omega\in\mathcal{D}_\alpha=\C^+\cup\R\backslash\mathrm{supp}\,\mu_\alpha$. We now distinguish the cases $t_-^\alpha\in[0,1)$ and $t_-^\alpha\in(-1,0)$.

For $t_-^\alpha\in[0,1)$, we conclude from ~\eqref{from the computational lemma} and \eqref{le diverging2} that there is $c'>0$ such that 
\begin{align}
T(\omega)\ge c'\begin{cases}
                                                                                  \frac{(\re\omega-E_-^\alpha)^{t_-^\alpha}}{\im \omega (\log \im\omega)^2}\,, &\textrm{ if }\qquad \re\omega-E_-^\alpha>\im \omega\,,\\
                                                                                   \frac{|\re\omega-E_-^\alpha|^{t_-^\alpha-1}}{|\log (E_-^\alpha-\re\omega)|^2}\,, &\textrm{ if }\qquad \re\omega-E_-^\alpha< -\im\omega\,,\\
                                                                                    \frac{(\im \omega)^{t_-^\alpha-1}}{(\log \im\omega)^2}\,, &\textrm{ if }\qquad \im\omega\ge|\re\omega-E_-^\alpha|\,,
                                                                                   \end{cases}
\end{align}
when $|\omega-E_-^\alpha|\le \delta$, and hence  $T(\omega)\nearrow \infty$ as $\omega\rightarrow E'\in[E_-^\alpha,E_-^\alpha+\delta)$ in $\mathcal{D}_\alpha$, since $t_-^\alpha<1$.

For $t_-^\alpha\in(-1,0)$, we conclude from ~\eqref{from the computational lemma} and \eqref{le diverging3} that there is $c'>0$ such that
\begin{align}
T(\omega)\ge c'\begin{cases}\frac{(\re \omega-E_-^\alpha)^{-t_-^\alpha}}{|\log\im\omega|^2\,\im \omega}\,, &\textrm{ if }\qquad \re\omega-E_-^\alpha>\im \omega\,,\\
                                                                                  \frac{ |\re\omega-E_-^\alpha|^{-t_-^\alpha-1}}{|\log (E_-^\alpha-\re\omega)|^2}\,, &\textrm{ if }\qquad \re\omega-E_-^\alpha<-\im\omega\,,\\
                                                                                    \frac{(\im \omega)^{-t_-^\alpha-1}}{|\log \im\omega|^2}\,, &\textrm{ if }\qquad \im\omega\ge|\re\omega-E_-^\alpha|\,,
                                                                                   \end{cases}
\end{align}
when $|\omega-E_-^\alpha|\le \delta$, and hence  $T(\omega)\nearrow \infty$ as $\omega\rightarrow E'\in[E_-^\alpha,E_-^\alpha+\delta)$, since $t_-^\alpha\in(-1,0)$.

The same argument shows, as $t_+^\alpha\in(-1,1)$, that $\widehat I_\alpha(\omega)$ diverges when $\omega$ approaches $(E_+^\alpha-\delta,E_+^\alpha]$ in $\mathcal{D	}_\alpha$.

Next, we know that the density $\rho_\alpha$ of $\mu_\alpha$ is  a.e.\ positive on $[E_-^\alpha+\delta',E_+^\alpha-\delta']$, with $0<\delta'\le \delta$. Hence there is a constant $c>0$, depending on $\delta'$, such that
\begin{align}
 \int_{\R}\frac{\dd\mu_\alpha(x)}{|x-\omega|^2}\ge c\,\frac{1}{\im \omega}\,,
\end{align}
for all $\omega\in\mathcal{D}_\alpha$ satisfying $\mathrm{dist}(\omega,[E_-^\alpha+\delta,E_+^\alpha-\delta])\le \delta'/2$. On the other hand, as the density~$\rho_\alpha$ is a.e.\ finite in the bulk, there  are constants $c$ and $C$, depending on $\delta'$, such that
\begin{align}
  \Big|\int_\R \frac{\dd \mu_\alpha(x)}{x-\omega}\Big|\le C+c\int_{E_-^\alpha+\delta}^{E_+^\alpha-\delta}\frac{\dd x}{|x-\omega|} \le C+c\,|\log\im \omega|\,,
\end{align}
for all $\omega\in\mathcal{D}_\alpha$ with $\mathrm{dist}(\omega,[E_-^\alpha+\delta,E_+^\alpha-\delta])\le \delta'/2$. Thus, from~\eqref{mr t}, we conclude that $T(\omega)$ diverges when $\omega$ tends to any $E\in[E_-^\alpha+\delta,E_+^\alpha-\delta]$ within $\mathcal{D}_\alpha$. Hence so must $\widehat I_\alpha(\omega)$.

However, since $\widehat I_\alpha(\omega_\beta(z))$ is uniformly bounded for all $z\in\mathcal{E}$ by~\eqref{det tredje}, we conclude that $\omega=\omega_\beta(z)$ must be separated outside $\mathrm{supp}\,\mu_\alpha$. That is $d_\alpha(\omega_\beta(z))$ must, by the continuity of $\omega_\beta$, be uniformly bounded from below on $\mathcal{E}$.

Once~\eqref{le gap} has been established,~\eqref{gap bounds everything} follows directly from the definitions of $I_\alpha$ and $I_\beta$ in~\eqref{le I ohne hat}.
\end{proof}

In sum, we have proved so far that there is a constant $C\ge 1$, depending only on $\mu_\alpha$ and~$\mu_\beta$ via their control parameters, such that
\begin{align}\label{Ialpha}
 C^{-1}\le I_\alpha(\omega_\beta(z))\le C\,,\qquad C^{-1}\le I_\beta(\omega_\alpha(z))\le C\,,\qquad\qquad z\in\mathcal{E}\,.
\end{align}
Thence, recalling~\eqref{le to be noted}, we observe that the imaginary parts of $m(z)\equiv m_{\mu_\alpha\boxplus\mu_\beta}(z)$, $\omega_\alpha(z)$ and $\omega_\beta(z)$ are all comparable
 on the domain $\mathcal{E}$, which proves~\eqref{le links} for some constant $C\ge 1$ depending only on~$\mu_\alpha$ and $\mu_\beta$. This concludes the proof of Proposition~\ref{le pro links}.

\section{Characterization of (regular) edges and Proof of Theorem~\ref{le theorem 1}} \label{section 4}

Recall that $\rho(x)$ denotes the (continuous) density function at $x\in\R$ of the free additive convolution measure $\mu_\alpha\boxplus\mu_\beta$ and that $m(z)$ is its Stieltjes transform at $z\in\C^+$. From~\eqref{le to be noted} and~\eqref{Ialpha} we note that $\im\omega_\alpha(E+\ii\eta)$ and $\im\omega_\beta(E+\ii\eta)$ vanish in the limit $\eta\searrow 0$ if $\rho(E)=0$. The next result shows that the ratio of the imaginary parts of the subordination functions has a finite and positive limit as the spectral parameter approaches the real line.

\begin{lem}\label{le limit along the imaginary axis lem} Suppose that $\mu_\alpha$ and $\mu_\beta$ satisfy Assumption~\ref{a.regularity of the measures}. Let $E\in\mathcal{J}$; see~\eqref{le J}. Then
 \begin{align}\label{le limit along the imaginary axis}
  \lim_{\eta\searrow 0}\frac{\im\omega_\alpha(E+\ii\eta)}{\im\omega_\beta(E+\ii\eta)}=\frac{I_\alpha(\omega_\beta(E))}{I_{\beta}(\omega_\alpha(E))}\,.
 \end{align}
 The limit, which is bounded from above and from below by strictly positive constants
 (see \eqref{Ialpha}), is a continuous function in $E$. The corresponding statements hold with the roles of the indices $\alpha$ and $\beta$ interchanged.
\end{lem}
\begin{rem}
 In~\eqref{le limit along the imaginary axis} we take the limit $z=E+i\eta\to E$  in the
  direction of the imaginary axes, yet as $\omega_\alpha$ and $\omega_\beta$ extend continuously to the real axis by~Proposition~\ref{prop extension}, we obtain the same limit when taken along any non-tangential direction to $\R$. 
\end{rem}
\begin{proof}[Proof of Lemma~\ref{le limit along the imaginary axis lem}] 
For $z\in\mathcal{E}\backslash \mathcal{J}$,  we have $ \im m(z)\ne 0$, thus from~\eqref{le to be noted},
\begin{align}
 \im \omega_\alpha(z)=\frac{\im m(z)}{I_\beta({\omega_\alpha(z)})}\,,\qquad\quad\im \omega_\beta(z)=\frac{\im m(z)}{I_\alpha({\omega_\beta(z)})}\,,
\end{align}
hence
\begin{align}\label{le haifisch}
 \frac{\im \omega_\alpha(E+\ii\eta)}{\im \omega_\beta(E+\ii\eta)}=\frac{I_\alpha(\omega_\beta(E+\ii\eta))}{I_\beta(\omega_\alpha(E+\ii\eta))}=\frac{\int_\R\frac{\dd\mu_\alpha(y)}{|y-\omega_\beta(E+\ii\eta)|^2}}{\int_\R\frac{\dd\mu_\beta(y)}{|y-\omega_\alpha(E+\ii\eta)|^2}}\,,\qquad \qquad \eta>0\,.
\end{align}
 From Lemma~\ref{gap corollary} we know that $\omega_\alpha(z)$ and $\omega_\beta(z)$ stay away from the support of the measures~$\mu_\beta$, respectively $\mu_\alpha$ for all $z\in\mathcal{E}$, by the continuity of $\omega_\alpha$ and $\omega_\beta$ and dominated convergence, we can take the limit $\eta\searrow 0$ in~\eqref{le haifisch} and conclude that the limit is a finite strictly positive number by~\eqref{det forsta} and Lemma~\ref{gap corollary}. Continuity of the limit is immediate.
 \end{proof}

   We are now ready to characterize the (regular) edges of the measure $\mu_\alpha\boxplus\mu_\beta$. For brevity, we use the following notation: Denote the set of vanishing points, $\mathcal{V}$, of $\mu_\alpha\boxplus\mu_\beta$ by
\begin{align}\label{le cal B}
 \mathcal{V}\deq\partial\{x\in\R\,:\,\rho(x)>0\}\,,
\end{align}
where $\rho$ denotes the density of $\mu_\alpha\boxplus\mu_\beta$. We remark that $\mathcal{V}$ is not  necessarily the boundary of the support of
$\mu_\alpha\boxplus\mu_\beta$. It may happen that $\mathcal{V}$ contains isolated zeros: $x\in\R$ is 
 called an {\it isolated zero} if $\rho(x)=0$, and $\rho(x+\epsilon)>0$ and $\rho(x-\epsilon)>0$, for all $\epsilon\in(0,\epsilon_0)$, for some $\epsilon_0>0$. 
 However, we will prove below
in Proposition~\ref{prop. support is an interval} that $\mu_\alpha\boxplus\mu_\beta$ does not have isolated zeros under Assumption~\ref{a.regularity of the measures}.

\begin{pro}\label{prop. woises}
Suppose that $\mu_\alpha$ and $\mu_\beta$ satisfy Assumption~\ref{a.regularity of the measures}. Then we have
\begin{align}\label{le equal 1}
 \big|(F'_{\mu_\alpha}(\omega_\beta(z))-1)(F'_{\mu_\beta}(\omega_\alpha(z))-1)\big|\le 1\,,
\end{align}
for all $z\in\C^+\cup\R$; see~\eqref{le domain}. Moreover, we have equality in~\eqref{le equal 1} with $z=E+\ii\eta\in\mathcal{E}$ if and only if the spectral parameter $z$ satisfies
\begin{align}\label{le edge ta}
 E\in\mathcal{V}\,,\qquad\quad \eta=0\,.
\end{align}
In fact, for such $E$, we have that $(F'_{\mu_\alpha}(\omega_\beta(E))-1)(F'_{\mu_\beta}(\omega_\alpha(E))-1)=1$.
\end{pro}
\begin{proof}
Inequality~\eqref{le equal 1} was proved in~\cite{BB} using the concept of Denjoy--Wolff points. Here we give an elementary direct argument. Recalling~\eqref{we can do what we want}, we indeed note that
 \begin{align}\label{big one}
  \big|(F'_{\mu_\alpha}(\omega_\beta(z))-1)(F'_{\mu_\beta}(\omega_\alpha(z))-1)\big|&=\Big|\int_\R\frac{\dd\widehat\mu_\alpha(x)}{(x-\omega_\beta(z))^2}\int_\R\frac{\dd\widehat\mu_\beta(x)}{(x-\omega_\alpha(z))^2}\Big|\nonumber\\
  &\stackrel{(i)}{ \le } \int_\R\frac{\dd\widehat\mu_\alpha(x)}{|x-\omega_\beta(z)|^2}\int_\R\frac{\dd\widehat\mu_\beta(x)}{|x-\omega_\alpha(z)|^2}\nonumber\\
  &=\widehat I_\alpha(\omega_\beta(z))\widehat I_\beta(\omega_\alpha(z))\nonumber\\
  &\mathrel{\stackrel{(ii)}{ \le}}1\,,
 \end{align}
 for all $z\in\mathcal{E}$, where we used~\eqref{le I hat}
 and \eqref{smaller than one}. We are now interested in the case when we have equality in~\eqref{le equal 1}.

 Next, assuming equality in $(i)$ and $(ii)$, we show that~\eqref{le edge ta} holds. First, recall from~\eqref{same support} that the measures $\mu_\alpha$ and $\widehat\mu_\alpha$ have the same support. Since $\mu_\alpha$ is supported on an interval, so must $\widehat\mu_\alpha$ be. Similarly for $\widehat\mu_\beta$. Since the subordination functions are uniformly bounded on~$\mathcal{E}$, we get equality in $(i)$ of~\eqref{big one} if and only if $z$ is such that $\im \omega_\alpha(z)=\im\omega_\beta(z)=0$. This entails, as $\im\omega_\alpha(z)\ge\im z$, $\im\omega_\beta(z)\ge\im z$, $z\in\C^+$ by Proposition~\ref{le prop 1} (and continuous extension to the real line), that such a $z\in\mathcal{E}$ must lie on the real line, \ie $\eta=0$.

To get the first part of~\eqref{le edge ta}, we note that from the definition of $\widehat I_\alpha$ and $\widehat I_\beta$  in~\eqref{le I hat} and~\eqref{le super trup}, we have
\begin{align}
\widehat I_\alpha(\omega_\beta(z))=\frac{\im \omega_\alpha(z)}{\im\omega_\beta(z)}-\frac{\im z}{\im\omega_\beta(z)}=\frac{I_\alpha(\omega_\beta(z))}{I_\beta(\omega_\alpha(z))}-\frac{\im z}{\im\omega_\beta(z)}\,,
\end{align}
where we used~\eqref{le haifisch} to get the second equality, and similarly with the roles of $\alpha$ and $\beta$ interchanged. Thus with $z=x+\ii\eta_0$, $x\in \R$, we can take, by Lemma~\ref{le limit along the imaginary axis lem}, $\eta_0$ to zero to get
\begin{align}
 \widehat I_\alpha(\omega_\beta(x))=\frac{I_\alpha(\omega_\beta(x))}{I_\beta(\omega_\alpha(x))}-\lim_{\eta_0\searrow 0}\frac{\eta_0}{\im\omega_\beta(x+\ii\eta_0)}\,.
\end{align}
Since $\widehat I_\alpha(\omega_\beta(x))>0$, the right side is strictly positive as well.

Let now $x=E$ be such that we have equality in $(i)$ and $(ii)$ in \eqref{big one}. Then we have
\begin{align*}
 1&=\widehat I_\alpha(\omega_\beta(E))\widehat I_\beta(\omega_\alpha(E))\\
 &=\left(\frac{I_\alpha(\omega_\beta(E))}{I_\beta(\omega_\alpha(E))}-\lim_{\eta_0\searrow 0}\frac{\eta_0}{\im\omega_\beta(E+\ii\eta_0)}\right)\left(\frac{I_\beta(\omega_\alpha(E))}{I_\alpha(\omega_\beta(E))}-\lim_{\eta_0\searrow 0}\frac{\eta_0}{\im\omega_\alpha(E+\ii\eta_0)}\right)
\end{align*}
We therefore conclude as both factors in the above product are positive that 
\begin{align}\label{le limit behavior}
 \lim_{\eta_0\searrow 0}\frac{\eta_0}{\im\omega_\beta(E+\ii\eta_0)}=\lim_{\eta_0\searrow 0}\frac{\eta_0}{\im\omega_\beta(E+\ii\eta_0)}=0\,.
\end{align}

Summarizing, we have so far proved that if there is $z=E+\ii\eta\in\mathcal{E}$ such that we have equality in $(i)$ and $(ii)$, then $\eta=0$, $\im \omega_\alpha(E)=\im\omega_\beta(E)=0$, and~\eqref{le limit behavior} hold.

It remains to show that such $E$ belongs to $\mathcal{V}$. From~\eqref{le to be noted}, we see that $\im\omega_\alpha(E)=0$ implies $\im m(E)=0$, \ie $E\in\{x\in\R\,:\,\rho(x)=0\}$. Moreover, from~\eqref{le limit behavior}, we have that 
\begin{align}
 \lim_{\eta_0\searrow 0}\frac{\im\omega_\beta(E+\ii\eta_0)}{\eta_0}=\infty\,,
\end{align}
which in turn means by~\eqref{le to be noted} and by~\eqref{Ialpha} that
\begin{align}
 \lim_{\eta_0\searrow 0}\frac{\im m(E+\ii\eta_0)}{\eta_0}=\infty\,.
\end{align}
However, if $E$ were in the complement of the support of $ \rho(x)\dd x= \dd\mu_\alpha\boxplus\mu_\beta(x)$, then 
\begin{align}
 \frac{\im m(E+\ii\eta_0)}{\eta_0}=\int_\R\frac{\dd\rho(x)}{|x-E|^2+\eta_0^2}\,,
\end{align}
would remain bounded as $\eta_0\searrow 0$. Thus $E\in\mathrm{supp}\,\mu_\alpha\boxplus\mu_\beta\cap\{x\in\R\,:\,\rho(x)=0\}$, \ie $E\in\mathcal{V}$.

 We now prove the converse:~\eqref{le edge ta} implies equality in $(i)$ and $(ii)$ in \eqref{big one}.
 Since $\rho(E)=0$, we also have $\im m(E)=0$ and $\im\omega_\alpha(E)=\im\omega_\beta(E)=0$  by Proposition~\ref{le pro links}. This gives equality in~$(i)$. 

 Next, as $E\in\mathcal{V}$, there is $\epsilon_0>0$ such that $\rho(E-\epsilon)>0$ or $\rho(E+\epsilon)>0$, for all $0<\epsilon\le\epsilon_0$. Assume first that $\rho(E-\epsilon)>0$, $0<\epsilon\le\epsilon_0$. Then, for a fixed $\epsilon>0$, we have by~\eqref{le to be noted} and~\eqref{Ialpha} that
\begin{align}\label{le ha goes to zero}
 \lim_{\eta_0\searrow 0}\frac{\eta_0}{\im \omega_\alpha(E-\epsilon+\ii\eta_0)}=\lim_{\eta_0\searrow 0}\frac{\eta_0}{\im \omega_\beta(E-\epsilon+\ii\eta_0)}=0\,.
\end{align}
Thus, by~\eqref{le super trup}, for any $\eta_0>0$,
\begin{align*}
 \widehat I_\alpha(\omega_\beta(E-\epsilon+\ii\eta_0))\,\widehat I_\beta(\omega_\alpha(E-\epsilon+\ii\eta_0))&=\left(\frac{\im\omega_\alpha(E-\epsilon+\ii\eta_0)}{\im\omega_\beta(E-\epsilon+\ii\eta_0)}-\frac{\eta_0}{\im\omega_\beta(E-\epsilon+\ii\eta_0)}\right)\nonumber\\ &\,\quad\cdot\left(\frac{\im\omega_\beta(E-\epsilon+\ii\eta_0)}{\im\omega_\alpha(E-\epsilon+\ii\eta_0)}-\frac{\eta_0}{\im\omega_\alpha(E-\epsilon+\ii\eta_0)}\right)
\end{align*}
Taking the limit $\eta_0\searrow 0$, for fixed $\epsilon>0$, we find from~\eqref{le ha goes to zero} that
\begin{align}
 \widehat I_\alpha(\omega_\beta(E-\epsilon))\,\widehat I_\beta(\omega_\alpha(E-\epsilon))&=\frac{\im\omega_\alpha(E-\epsilon)}{\im\omega_\beta(E-\epsilon)}\,\frac{\im\omega_\beta(E-\epsilon)}{\im\omega_\alpha(E-\epsilon)}=1\,.
\end{align}

Next, using the continuity of $\omega_\alpha$ and $\omega_\beta$, and that they are separated from the support of $\mu_\beta$, respectively $\mu_\alpha$, we have by continuity that
\begin{align}
\widehat I_\alpha(\omega_\beta(E))\,\widehat I_\beta(\omega_\alpha(E))=\lim_{\epsilon\searrow 0}\widehat I_\alpha(\omega_\beta(E-\epsilon))\,\widehat I_\beta(\omega_\alpha(E-\epsilon)) =1\,,
\end{align}
and we obtain equality in~$(ii)$ assuming that $\rho(E-\epsilon)>0$, $0<\epsilon\le\epsilon_0$. The case $\rho(E+\epsilon)>0$, $0<\epsilon\le\epsilon_0$, is handled in the same way and we find that in either case we have equality in~$(ii)$. Thus we have proved that~\eqref{le edge ta} implies equality in $(i)$ and $(ii)$ in \eqref{big one}.

Finally, since $\omega_\alpha(E)$ and $\omega_\beta(E)$ are real valued for any $E\in\mathcal{V}$ (see~\eqref{le to be noted} and use~\eqref{Ialpha}), and since they are separated from the support of the respective measure $\mu_\beta$ and $\mu_\alpha$, we conclude the $(F'_{\mu_\alpha}(\omega_\beta(E))-1)(F'_{\mu_\beta}(\omega_\alpha(E))-1)$ is real valued, and hence by~\eqref{we can do what we want} positive. Thence we must have $(F'_{\mu_\alpha}(\omega_\beta(E))-1)(F'_{\mu_\beta}(\omega_\alpha(E))-1)=1$ for all $E\in\mathcal{V}$. This finishes the proof of Proposition~\ref{prop. woises}.
\end{proof} 

We need one more technical lemma before we can move on to the proof of Theorem~\ref{le theorem 1}.
\begin{lem}\label{le nu lemma} Let $\omega_\alpha$ and $\omega_\beta$ be the subordination functions associated to $\mu_\alpha$ and $\mu_\beta$ by Proposition~\ref{le prop 1}. Then there exist finite Borel measures $\nu_\alpha$ and $\nu_\beta$ on $\R$ such that
\begin{align}\label{le sa1}
 \omega_\alpha(z)-z=\int_\R\frac{\dd\nu_\alpha(x)}{x-z}\,,\qquad \omega_\beta(z)-z=\int_\R\frac{\dd\nu_\beta(x)}{x-z}\,,
\end{align} 
for any $z$ outside the corresponding supports, where
\begin{align}\label{le sa2}
 0<\nu_\alpha(\R)=\int_\R x^2\dd\mu_\alpha(x)\,,\qquad 0<\nu_\beta(\R)=\int_\R x^2\dd\mu_\beta(x)\,.
\end{align}
Moreover we have
\begin{align}\label{le sa3}
 \mathrm{supp}\,\nu_\alpha=\mathrm{supp}\,\nu_\beta=\mathrm{supp}\,\mu_\alpha\boxplus\mu_\beta\,,
\end{align}
thus \eqref{le sa1} holds for any $z\in \C\setminus \mathrm{supp}\,\mu_\alpha\boxplus\mu_\beta$.

\end{lem}
\begin{rem}
 The measure $\nu_\alpha$ is referred to as the subordination distribution of $\mu_\alpha\boxplus\mu_\beta$ with respect to $\mu_\beta$. Analogously, $\nu_\beta$ is the subordination distribution of $\mu_\alpha\boxplus\mu_\beta$ with respect to~$\mu_\alpha$; see~\cite{Lenc}. Statement~\eqref{le sa2} is a particular simple case  of the results in~\cite{Lenc}; see also Theorem~1.2 in~\cite{Nica}. For completeness we include an elementary proof of~\eqref{le sa2}.
\end{rem}

\begin{proof}[Proof of Lemma~\ref{le nu lemma}]
 The proof is similar to the proof of Lemma~\ref{lemm on hat measures}. We start from $\omega_\alpha(\ii\eta)+\omega_\beta(\ii\eta)-\ii\eta=F_{\mu_\alpha}(\omega_\beta(\ii\eta))$, $\eta>0$. By~\eqref{le limit of omega} we have $\lim_{\eta\nearrow\infty}\omega_\alpha(\ii\eta)/\ii\eta=1$ and we can expand $F_{\mu_\alpha}(\omega_\beta(\ii\eta))$ around infinity similar to~\eqref{used to compare}. On the other 
 hand as $z\mapsto\omega_\alpha(z)-z$ is a self-map of the upper half plane by Proposition~\ref{le prop 1} it admits the representation~\eqref{le pick} with a measure $\nu_\alpha\deq\mu$. By comparison we directly find~\eqref{le sa1} and~\eqref{le sa2} for $\omega_\alpha$. For $\omega_\beta$ these equations are established in the same way.
 
 Finally by~\eqref{le sa1}, the functions $\omega_\alpha(z)$ and $\omega_\beta(z)$ are analytic outside $\mathrm{supp}\,\nu_\alpha$, respectively $\mathrm{supp}\,\nu_\beta$, and $\im \omega_\alpha(E)=0$, for $E\in\R\backslash\mathrm{supp}\,\nu_\alpha$, and $\im\omega_\beta(E)=0$, for $E\in\R\backslash\mathrm{supp}\,\nu_\beta$. We have established in Lemma~\ref{gap corollary} that $\omega_\alpha(z)$ and $\omega_\beta(z)$, $z\in\mathcal{E}$, stay away from the supports of the measures~$\mu_\beta$, respectively~$\mu_\alpha$. Thus $m_{\mu_\alpha}(\omega_\beta(z))$ and $m_{\mu_\beta}(\omega_\alpha(z))$ are analytic outside $\mathrm{supp}\,\nu_\alpha$, respectively $\mathrm{supp}\,\nu_\beta$, with $\im m_{\mu_\alpha}(\omega_\beta(E))=0$, for $E\in\R\backslash\mathrm{supp}\,\nu_\beta$, 
and $\im m_{\mu_\beta}(\omega_\alpha(E))=0$, for $E\in\R\backslash\mathrm{supp}\,\nu_\alpha$. By analytic subordination we have $m(z)=m_{\mu_\alpha}(\omega_\beta(z))=m_{\mu_\beta}(\omega_\alpha(z))$, $z\in\C^+$. Thus, since the subordination functions continuously extend to the real line, we have $\im m(E)=\im m_{\mu_\alpha}(\omega_\beta(E))=\im m_{\mu_\beta}(\omega_\alpha(E))=0$ and we conclude that $\mathrm{supp}\,\nu_\alpha=\mathrm{supp}\,\nu_\beta$ as well as $\mathrm{supp}\,\mu_\alpha\boxplus\mu_\beta\subseteq\mathrm{supp}\,\nu_\alpha$. Finally, let $E\in\R\backslash\mathrm{supp}\,\mu_\alpha\boxplus\mu_\beta$. Then $m(z)$ is analytic in a neighborhood of $E$ and we have $\im m(E)=\im m_{\mu_\alpha}(\omega_\beta(E))=0$. Since by Lemma~\ref{gap corollary} $\omega_\beta(E)$ is outside $\mathrm{supp}\,\mu_\alpha$, we also have $\im\omega_\beta(E)=0$. Recalling~\eqref{le sa1}, we conclude that $E\not\in\mathrm{supp}\,\nu_\beta$.  Thus we have $\mathrm{supp}\,\nu_\alpha\subseteq\mathrm{supp}\,\mu_\alpha\boxplus\mu_\beta$ and we conclude that~\eqref{le sa3} holds.
 \end{proof}
 
 \begin{rem}
  The subordination functions extend continuously to $\R$ and they are real analytic outside the support of $\mu_\alpha\boxplus\mu_\beta$. As $\nu_\alpha$ and $\nu_\beta$ are finite measures by~\eqref{le sa2}, dominated convergence asserts that
  \begin{align}\label{le omegas are increasing}
   \omega'_\alpha(E)-1=\int_\R\frac{\dd\nu_\alpha(x)}{(E-x)^2}\,,\quad\omega_\beta'(E)-1=\int_\R\frac{\dd\nu_\beta(x)}{(E-x)^2}\,,\qquad E\in\R\backslash\mathrm{supp}\mu_\alpha\boxplus\mu_\beta\,.
  \end{align}
In particular, the subordination functions are strictly increasing in $E$ on $\R\backslash\mathrm{supp}\mu_\alpha\boxplus\mu_\beta$. Moreover from~\eqref{le sa1}, we have $\lim_{E\rightarrow\pm\infty}\omega_\alpha(E)=\pm\infty$ and the same holds true for $\omega_\beta$.
 \end{rem}

Having established Proposition~\ref{prop. woises}, we are now ready to determine the support of the free convolution measure. Recall the notation $\mathcal{V}=\partial\{ x\in\R\,:\, \rho(x)>0\}$ from~\eqref{le cal B}. 
\begin{pro}\label{prop. support is an interval}
 
Suppose that $\mu_\alpha$ and $\mu_\beta$ satisfy Assumption~\ref{a.regularity of the measures}. Then there exist finite numbers $E_-<0$ and $E_+>0$ such that $\mathcal{V}=\{E_-,E_+\}$ and 
\begin{align}\label{le support is an interval equation}
 \{x\in\R\,:\,\rho(x)>0\}=(E_-,E_+)\,.
\end{align}
In particular we have $\mathrm{supp}\,\mu_\alpha\boxplus\mu_\beta=[E_-,E_+]$.
\end{pro}

\begin{proof}[Proof of Proposition~\ref{prop. support is an interval}]
 From Proposition~\ref{prop. woises} we know that a point $E$ belongs to $\mathcal{V}$ if and only if $(F'_{\mu_\alpha}(\omega_\beta(E))-1)(F'_{\mu_\beta}(\omega_\alpha(E))-1)=1$. Using~\eqref{we can do what we want}, we rewrite this condition as 
 \begin{align}\label{le find the zeros}
  \widehat I_\alpha(\omega_\beta(E))\widehat I_\beta(\omega_\alpha(E))=\int_\R\frac{\dd\widehat\mu_\alpha(y)}{(y-\omega_\beta(E))^2}\int_\R\frac{\dd\widehat\mu_\beta(y)}{(y-\omega_\alpha(E))^2}=1\,.
 \end{align}
 
We will now look for solutions to~\eqref{le find the zeros} for $E\in\mathcal{J}$. For ease of notation set
\begin{align}
 f(E)\deq \int_\R\frac{\dd\widehat\mu_\alpha(y)}{(y-\omega_\beta(E))^2}\int_\R\frac{\dd\widehat\mu_\beta(y)}{(y-\omega_\alpha(E))^2}\,,\qquad\qquad E\in\R\,.
 \end{align}
By~\eqref{le omegas are increasing} we know that the subordination functions are strictly increasing outside the support of $\mu_\alpha\boxplus\mu_\beta$ and we also know from~\eqref{le sa1} that $\omega_\alpha(E)=E + o(1)$ and $\omega_\beta(E)=E + o(1)$, as $E\searrow-\infty$. Hence $f(E)=o(1)$ as $E\searrow -\infty$. We now increase $E$ starting from $-\infty$ and note that $f(E)$ is  monotone increasing, as $\omega_\alpha(E)$ and $\omega_\beta(E)$ are.

By~\eqref{le equal 1}, we know that $|f(E)|\le 1$, for all $E\in\R$. Yet, we also know from the paragraph below~\eqref{le diverging2} in the proof of Lemma~\ref{gap corollary} that $\widehat I_\alpha(\omega)$ and $\widehat I_\beta(\omega)$ both diverge when $\omega$ approaches the lower endpoints of the measures $\mu_\alpha$, respectively $\mu_\beta$.  We therefore conclude by monotonicity 
of $f$ that there is only one solution, $E_-$, to~\eqref{le find the zeros} such that $\omega_\alpha(E_-)<E_-^\beta$ and $\omega_\beta(E_-)<E_-^\alpha$.  
The point $E_-$ is the first point from the left reaching  the support of $\mu_\alpha\boxplus\mu_\beta$, \ie it is the leftmost endpoint,
 and we therefore conclude that $E_-\in\mathcal{J}$ as $\mathrm{supp}\,\mu_\alpha\boxplus\mu_\beta\subset\mathcal{J}$;
 \cf remark below~\eqref{le domain}.

The same reasoning, reducing $E$  from $\infty$, shows that there is only one solution, $E_+$, to~\eqref{le find the zeros} such that $\omega_\alpha(E_+)>E_+^\beta$ and $\omega_\alpha(E_+)>E_-^\alpha$. Moreover, $E_+$ must be the right most endpoint of the support of $\mu_\alpha\boxplus\mu_\beta$ and hence $E_+\in\mathcal{J}$.

Any other solution, $E'$, to~\eqref{le find the zeros} must lie in $(E_-,E_+)$ and has to either satisfy 
\begin{align}\label{le condi1}
\omega_\alpha(E')<E_-^\beta\,\quad\textrm{and}\quad\omega_\beta(E')>E_+^\alpha\,,
\end{align}
or
\begin{align}\label{le condi2}
\omega_\alpha(E')>E_+^\beta\,\quad\textrm{and}\quad\omega_\beta(E')<E_-^\alpha\,. 
\end{align}
Yet, we now show that this cannot happen for $\omega_\alpha$ and $\omega_\beta$ solutions to the subordination equations~\eqref{le definiting equations}. 

Arguing by contradiction, assume there is a solution, $E'$, to~\eqref{le find the zeros} with $E_-<E'<E_+$. Then $E'$ satisfies either~\eqref{le condi1} or~\eqref{le condi2}. Assume first that~\eqref{le condi1} is satisfied. Since~$\mu_\beta$ is supported on a single interval, we must have $F_{\mu_\beta}(\omega_\alpha(E'))<0$. This can for example be seen from the representation
\begin{align}
 \frac{-1}{F_{\mu_\beta}(\omega)}=m_{\mu_\beta}(\omega)=\int_{E_-^\beta}^{E_+^\beta}\frac{\dd\mu_\beta(x)}{x-\omega}\,,\qquad\omega\in\R\,,
\end{align}
and the observation that $m_{\mu_\beta}(\omega)$ is strictly positive on $(-\infty,E_-^\beta)$.
 But, on the other hand, since $E_+^\alpha<\omega_\beta(E')$, we must have
$0<F_{\mu_\alpha}(\omega_\beta(E'))$ since $\mu_\alpha$ is supported on a single interval. Hence, we must have
\begin{align}\label{ywc}
  F_{\mu_\beta}(\omega_\alpha(E'))<0<F_{\mu_\alpha}(\omega_\beta(E'))\,.
\end{align}
However, by subordination we have $F_{\mu_\beta}(\omega_\alpha(E))=F_{\mu_\alpha}(\omega_\beta(E))$, for all $E\in\R$, contradicting~\eqref{ywc}. We therefore conclude that there is no solution $E'$ to~\eqref{le find the zeros} satisfying~\eqref{le condi1}.

The same argument shows that there cannot be a solution $E'$ to~\eqref{le find the zeros} satisfying~\eqref{le condi2}. Hence the only solutions to~\eqref{le find the zeros} are $E_-$ and $E_+$. We thus have $\mathcal{V}=\{E_-,E_+\}$, so that $\{x\in\R\,:\,\rho(x)>0\}=(E_-,E_+)$. In particular, $\mathrm{supp}\,\mu_\alpha\boxplus\mu_\beta=[E_-,E_+]$ and there are no isolated zeros. This concludes the proof of Proposition~\ref{prop. support is an interval}.
\end{proof}

\begin{pro}\label{prop. no cusps}
Suppose that $\mu_\alpha$ and $\mu_\beta$ satisfy Assumption~\ref{a.regularity of the measures}, in particular the measures $\mu_\alpha$ and $\mu_\beta$ are both supported on a single interval. Let the single interval support of $\mu_\alpha\boxplus \mu_\beta$ be denoted by $[E_-, E_+]$ as  in  Proposition~\ref{prop. support is an interval}. Then there are strictly positive constants $\gamma_-^\beta$ and $\gamma_+^\beta$ such that 
\begin{align}\label{le lower edge srt}
 \omega_\beta(z)=\omega_\beta(E_-)+\gamma_-^\beta\sqrt{E_--z}+O(|z-E_-|)\,,
\end{align}
 for $z$ in a neighborhood of $E_-$, where we choose the square root such that $\im\omega_\beta(x)>0$, $x> E_-$. Similarly, we have
 \begin{align}\label{le upper edge srt}
 \omega_\beta(z)=\omega_\beta(E_+)+\gamma_+^\beta\sqrt{z-E_+}+O(|z-E_+|)\,,
\end{align}
for $z$ in a neighborhood of $E_+$, where we choose the square root such that $\im\omega_\beta(x)>0$, $x<E_+$. The same conclusions apply to $\omega_\alpha$ with strictly positive constants $\gamma_-^\alpha$ and $\gamma_+^\alpha$.
\end{pro}

\begin{rem}
 
 The proof of Proposition~\ref{prop. no cusps} follows a similar strategy as the proof of Lemma~3.7 in~\cite{BES17}. Theorem~\ref{le theorem 1} will be a direct consequence of Proposition~\ref{prop. no cusps} and the subordination equations.
\end{rem}

\begin{proof}[Proof of Proposition~\ref{prop. no cusps} ]
We focus on the lower edge $E_-$ and prove~\eqref{le lower edge srt}. Equation~\eqref{le upper edge srt} is proved in the analogous way. We start by rewriting the subordination equation~\eqref{le definiting equations} in the form of a fixed point equation. Using Lemma~\ref{gap corollary}, and the fact that $|F'_{\mu_\beta}(\omega)|>0$, $\omega\in\R\backslash\mathrm{supp}\,\mu_\beta$, as follows from~\eqref{we can do what we want}, we conclude by the analytic inverse function theorem that the functional inverse  $F^{(-1)}_{\mu_\beta}$ of $F_{\mu_\beta}$ is analytic in a neighborhood of $F_{\mu_\beta}(\omega_\alpha(E_-))$. Thus the function
\begin{align}\label{le z}
 \widetilde z(\omega)\deq -F_{\mu_\alpha}(\omega)+\omega+F^{(-1)}_{\mu_\beta}\circ F_{\mu_\alpha}(\omega)\,,
\end{align}
is well-defined and analytic in a neighborhood of $\omega_\beta(E_-)$. It further follows from~\eqref{le definiting equations} that $\omega_\beta(z)$ is a solution $\omega=\omega_\beta(z)$ to the equation $z=\widetilde z(\omega)$ (with $\im \omega_\beta(z)\ge \im z$). Moreover, we have $\omega_\alpha(z)=F^{(-1)}_{\mu_\beta}\circ F_{\mu_\alpha}(\omega_\beta(z))$.

As argued in the proof of Proposition~\ref{prop. support is an interval}, the lower edge $E_-$ satisfies
\begin{align}\label{the new guy}
 (F'_{\mu_\alpha}(\omega_\beta(E'))-1)(F'_{\mu_\beta}(\omega_\alpha(E'))-1)=1\,,
\end{align}
as well as $\omega_\alpha(E_-)<E_-^\beta$ and $\omega_\beta(E_-)<E_-^\alpha$. The function $\widetilde{z}(\omega)$ can then be analytically continued to a neighborhood of $\omega_\beta(E_-)$ with Taylor expansion
\begin{multline}\label{le Taylor expansion}
\widetilde{z}(\omega)=E_-+\widetilde z'(\omega_\beta(E_-))(\omega-\omega_\beta(E_-))+\frac{1}{2}\widetilde z''(\omega_\beta(E_-))(\omega-\omega_\beta(E_-))^2\\  +O\left((\omega-\omega_\beta(E_-))^3\right)\,.
\end{multline}
We compute from~\eqref{le z} that
\begin{align}\label{le z prime}
 \widetilde z'(\omega)=-F'_{\mu_\alpha}(\omega)+1+\frac{1}{F'_{\mu_\beta}\circ F_{\mu_\beta}^{(-1)}\circ F_{\mu_\alpha}(\omega)}F'_{\mu_\alpha}(\omega)\,.
\end{align}
It is straightforward to check that $\widetilde z'(\omega_\beta(E_-))=0$ as $E_-$ is a solution to~\eqref{the new guy}. Yet, we claim that $\widetilde z''(\omega_\beta(E'))<0$. From~\eqref{le z prime} we compute,
\begin{align*}
 \widetilde z''(\omega)&=-F''_{\mu_\alpha}(\omega)+\frac{1}{F'_{\mu_\beta}\circ F_{\mu_\beta}^{(-1)}\circ F_{\mu_\alpha}(\omega)}F''_{\mu_\alpha}(\omega)\nonumber\\&\qquad-\frac{1}{(F'_{\mu_\beta}\circ F_{\mu_\beta}^{(-1)}\circ F_{\mu_\alpha}(\omega))^3}\left(F''_{\mu_\beta}\circ F_{\mu_\beta}^{(-1)}\circ F_{\mu_\alpha}(\omega)\right)\cdot( F'_{\mu_\alpha}(\omega))^2\nonumber\,,
\end{align*}
and thus by choosing $\omega=\omega_\beta(z)$, we get
\begin{align}\label{jasz}
 \widetilde z''(\omega_\beta(z))=-\frac{F''_{\mu_\alpha}(\omega_\beta(z))}{F'_{\mu_\beta}(\omega_\alpha(z))}(F'_{\mu_\beta}(\omega_\alpha(z))-1)-\frac{F''_{\mu_\beta}(\omega_\alpha(z))}{(F'_{\mu_\beta}(\omega_\alpha(z)))^3}( F'_{\mu_\alpha}(\omega_\beta(z)))^2\,.
\end{align}
Next, recall~\eqref{we can do what we want} and that $\omega_\alpha(E_-)<E_-^\beta$ as well as $\omega_\beta(E_-)<E_-^\alpha$. Thus
\begin{align}
 F'_{\mu_\beta}(\omega_\alpha(E_-))>1\,,\quad F'_{\mu_\alpha}(\omega_\beta(E_-))>1\,,\quad F''_{\mu_\alpha}(\omega_\beta(E_-))>0\,, \quad F''_{\mu_\beta}(\omega_\alpha(E_-))>0\,,
\end{align}
which implies upon choosing $z=E_-$ in~\eqref{jasz} that
\begin{align}\label{flugi}
 c\le-\widetilde z''(\omega_\beta(E_-))\le C\,. 
\end{align}
Hence choosing $\omega=\omega_\beta(z)$ in the Taylor expansion of $\widetilde z(\omega)$ in~\eqref{le Taylor expansion} 
 (thus $\widetilde z(\omega_\beta(z))=z)$
and using  $\widetilde z'(\omega_\beta(E_-))=0$, $\widetilde z''(\omega_\beta(E_-))\not=0$, we get
\begin{align}\label{kind nervt}
 \omega_\beta(z)=\omega_\beta(E_-)+\sqrt{\frac{-2}{z''(\omega_\beta(E_-))}}\sqrt{E_--z}+O(|E_--z|)\,,
\end{align}
 for $z$ in a neighborhood of $E_-$, where we choose the square root such that $\im\omega_\beta(x)>0$, $x> E_-$. Choosing $\gamma_-^\beta\deq(-2/z''(\omega_\beta(E_-)))^{1/2}$, we obtain~\eqref{le lower edge srt}.
 \end{proof}
Finally, we complete the proof of Theorem~\ref{le theorem 1}.
\begin{proof}[Proof of Theorem~\ref{le theorem 1}]
 It suffices to recall from~\eqref{le to be noted} that 
 \begin{align}
  \im m_{\mu_\alpha\boxplus\mu_\beta}(z)=\im\omega_\beta(z)\cdot I_{\alpha}(\omega_\beta(z))\,,\qquad z\in\C^+\,,
 \end{align}
 and that $I_{\alpha}(\omega_\beta(z))$ is, by~\eqref{det forsta} and~\eqref{gap bounds everything}, uniformly bounded from below and above for all $z\in\mathcal{E}$.
Theorem~\ref{le theorem 1} now directly follows from Proposition~\ref{prop. no cusps} and the Stieltjes inversion formula.
\end{proof}

\end{document}